\documentclass[11pt,3p]{elsarticle}

\usepackage{amsfonts}
\usepackage{amsthm}
\usepackage{amsmath}
\usepackage{times}
\usepackage{lineno}
\usepackage{multirow,array}

\newcommand{\nats}{{\mathbb N}}

\DeclareMathOperator{\Fac}{Fact}
\DeclareMathOperator{\alp}{\textrm{Alph}}
\DeclareMathOperator{\MF}{\textrm{MF}}
\renewcommand{\epsilon}{\varepsilon}

\newtheorem{theorem}{Theorem}
\newtheorem{proposition}[theorem]{Proposition}
\newtheorem{lemma}[theorem]{Lemma}
\newtheorem{corollary}[theorem]{Corollary}

\newtheorem{example}{Example}

\newtheorem{remark}{Remark}

\begin{document}

\sloppy


\begin{frontmatter}

\title{Cyclic Complexity of Words\tnoteref{note1}}
\tnotetext[note1]{Some of the results in this paper were presented at the 39th International Symposium on Mathematical Foundations of Computer Science, MFCS~2014 \cite{mfcs}.}

\author[label1]{Julien Cassaigne}
\ead{julien.cassaigne@math.cnrs.fr}

\author[label2]{Gabriele Fici\corref{cor1}}
\ead{gabriele.fici@unipa.it}

\author[label2]{Marinella Sciortino}
\ead{marinella.sciortino@unipa.it}

\author[label4]{Luca Q. Zamboni}
\ead{zamboni@math.univ-lyon1.fr}

\address[label1]{Institut de Math\'ematiques de Marseille, CNRS UMR 7373,
Site Sud, Campus de Luminy\\ F-13288 Marseille cedex 9, France}
\address[label2]{Dipartimento di Matematica e Informatica, Universit\`a di Palermo\\Via Archirafi 34, 90123 Palermo, Italy}
\address[label4]{Universit\'e Claude Bernard Lyon 1, CNRS UMR 5208, Institut Camille Jordan, 43 blvd. du 11 novembre 1918\\ F-69622 Villeurbanne cedex, France}

\cortext[cor1]{Corresponding author.}

\journal{Journal of Combinatorial Theory, Series A}

\begin{abstract}
We introduce and study a  complexity function on words $c_x(n),$  called \emph{cyclic complexity}, which counts the number of conjugacy classes of factors of length $n$ of an infinite word $x.$ We extend the well-known Morse-Hedlund theorem to the setting of cyclic complexity by showing that a word is ultimately periodic if and only if it has bounded cyclic complexity. Unlike most complexity functions, cyclic complexity distinguishes between Sturmian words of different slopes. We prove that if $x$ is a Sturmian word and $y$ is a word having the same cyclic complexity of $x,$ then up to renaming letters, $x$ and $y$ have the same set of factors. In particular, $y$ is also Sturmian  of slope equal to that of $x.$  Since $c_x(n)=1$ for some $n\geq 1$ implies $x$ is periodic,  it is natural to consider the quantity $\liminf_{n\rightarrow \infty} c_x(n).$   We show that if $x$ is a Sturmian word, then $\liminf_{n\rightarrow \infty} c_x(n)=2.$  We prove however that this is not a characterization of Sturmian words by exhibiting a restricted class of Toeplitz words, including the period-doubling word, which also verify this same condition on the limit infimum. In contrast we show that, for the Thue-Morse word $t$, $\liminf_{n\rightarrow \infty} c_t(n)=+\infty.$ 
\end{abstract}


\end{frontmatter}

\section{Introduction}

The factor complexity $p_x(n)$  of an infinite word $x=x_0x_1x_2\cdots \in A^\nats$ (with each $x_i$ belonging to a finite nonempty alphabet $A)$ counts the number of distinct factors $x_ix_{i+1}\cdots x_{i+n-1}$ of length $n$ occurring in $x.$ It provides a measure of the extent of randomness of the word $x$ and more generally of the subshift generated by $x.$ Periodic words have bounded factor complexity while digit expansions of normal numbers have full complexity. A celebrated result of Hedlund and Morse in \cite{MoHe38}  states that every non-periodic word  contains at least $n+1$ distinct factors of each length $n.$  Moreover, there exist words satisfying  $p_x(n)=n+1$ for each $n\geq 1.$ These words are called Sturmian words, and in terms of their factor complexity,  are regarded to be the simplest non-periodic words.

Sturmian words admit many different characterizations
of combinatorial, geometric and arithmetic nature. In the 1940's,  Hedlund and Morse showed that each Sturmian word is the symbolic  coding of the orbit of a point $x$ on the unit circle under a rotation by an irrational angle $\theta ,$ called the slope, where the circle is  partitioned into two complementary intervals, one of length $\theta $ and the
other of length $1-\theta .$ Conversely, each such coding defines a Sturmian word. It is well known that the dynamical/ergodic properties of the system, as well as the combinatorial properties of the associated Sturmian word, hinge on the arithmetical/Diophantine qualities of the slope $\theta$ given by its continued fraction expansion.
Sturmian words arise naturally in various branches of mathematics  including combinatorics, algebra,
number theory, ergodic theory, dynamical systems and differential equations.
They also have implications in theoretical physics as $1$-dimensional models of quasi-crystals. 

Other measures of complexity of words have been introduced and studied in the literature, including abelian complexity, maximal pattern complexity, $k$-abelian complexity, binomial complexity,  periodicity complexity, minimal forbidden factor complexity and palindromic complexity. With respect to most word complexity functions, Sturmian words are characterized as those non-periodic  words of lowest complexity.
One exception to this occurs in the context of maximal pattern complexity introduced by Kamae in \cite{KaZa}. In this case, while  Sturmian words are {\it pattern Sturmian}, meaning that they have minimal maximal pattern complexity amongst all non-periodic words, they are not the only ones. In fact, a certain restricted class of Toeplitz words which includes the period-doubling word are also known to be pattern Sturmian (see \cite{KaZa01}). On the other hand, the Thue-Morse word, while  closely connected to the period-doubling word, is known to have full maximal pattern complexity (see Example 1 in \cite{KaZa}).

In this paper we consider a new measure of complexity, \emph{cyclic complexity}, which consists in counting the factors of each given length of an infinite word  up to conjugacy. Two words $u$ and $v$ are said to be {\it conjugate}  if and only if  $u=w_{1}w_{2}$ and $v=w_{2}w_{1}$ for some words $w_{1},w_{2},$ i.e.,  if they are equal when read on a circle. The cyclic complexity of a word is the function which counts the number of conjugacy classes of factors of each given length. We note that factor complexity, abelian complexity and cyclic complexity can all be viewed as actions of different subgroups of the symmetric group on the indices of a finite word (respectively, the trivial subgroup, the whole symmetric group and the cyclic subgroup generated by the permutation $(1,2,\ldots ,n)).$

We establish the following analogue of the  Morse-Hedlund theorem:

\begin{theorem}\label{theor1}
A word $x$ is ultimately periodic if and only if it has bounded cyclic complexity.
\end{theorem}

The factor complexity does not distinguish between  Sturmian words of different slopes. In contrast, for cyclic complexity the situation is quite different. Indeed, we prove:

\begin{theorem}\label{theor2}
 Let $x$ be a Sturmian word. If  $y$ is an infinite word whose cyclic complexity is equal to that of  $x,$ then up to renaming letters,  $x$ and $y$ have the same set of factors. In particular, $y$ is also Sturmian.
\end{theorem}

A word is (purely) periodic if and only if there exists an integer $n$ such that all  factors of length $n$ are conjugate. Therefore, the minimum value of the cyclic complexity of a non-periodic word  is 2. We prove that if $x$ is a  Sturmian word then $\liminf_{n\rightarrow \infty}c_x(n)=2.$ We show however that this is not a characterization of Sturmian words by exhibiting a family of Toeplitz  words, which includes the period-doubling word, for which $\liminf_{n\rightarrow \infty}c_x(n)=2.$  We further show that if $x$ is a paperfolding word, then for every $n\geq 1$ one has $c_x(4\cdot 2^n)=4$. In contrast, we prove that for the Thue-Morse word, $\liminf_{n\rightarrow \infty}c_x(n)=+\infty.$

\section{Basics}\label{sec:def}

Given a finite nonempty ordered set $A$ (called the {\it alphabet}), we let $A^*$ and $A^\nats$ denote  respectively the set of finite words and the set of (right) infinite words over the alphabet $A$. The order on the alphabet $A$ can be extended to the usual lexicographic order on the set $A^{*}$.

For a finite word $w =w_1w_2\cdots w_n$ with $n \geq 1$ and $w_i \in A,$ the length $n$ of $w$ is denoted by $|w|.$ The  \textit{empty word} is denoted by $\varepsilon$ and we set $|\varepsilon|=0.$ We let $A^{n}$ denote the set of words of length $n$ and $ A^+$ the set of nonempty words. For $u,v \in A^+$,  $|u|_v$ is the number of occurrences of $v$ in $u.$
The \emph{Parikh vector} of $w$ is the vector whose components are the number of occurrences of the letters of $A$ in $w$. For example, if $A=\{a,b,c\}$, then the Parikh vector of $w=abb$ is $(1,2,0)$.
The \emph{reverse} (or \emph{mirror image}) of a finite word $w$ is the word $\tilde{w}$ obtained by reading $w$ in the reverse order.

Given a finite or  infinite word $\omega =\omega_1\omega_2\cdots $ with $\omega_i\in A,$ we say  that a word $u\in A^+$ is a  {\it factor} of $\omega$ if  $u=\omega_{i}\omega_{i+1}\cdots \omega_{i+|u|-1}$ for some $i\in \nats.$
We let $\Fac(\omega)$ denote the set of all factors of $\omega,$
and $\alp(\omega)$ the set of all factors of $\omega$ of length $1.$ If $\omega=u\nu$, we say that $u$ is a \emph{prefix} of $\omega$, while $\nu$ is a \emph{suffix} of $\omega$.
A factor $u$ of $\omega$ is called {\it right special} (resp.~{\it left special}) if both $ua$ and $ub$ (resp.~$au$ and $bu$) are factors of $\omega$ for distinct letters $a,b \in A.$ The factor $u$ is called \emph{bispecial} if it is both right special and left special. Furthermore, a bispecial factor $u$ of $\omega$ is \emph{strongly bispecial} if $aub\in \Fac(\omega)$ for every possible choice of $a$ and $b$ in $A$.

For each factor $u$ of $\omega$, we  set
\[\omega\big|_{u}=\{ n\in \nats \mid \omega_n\omega_{n+1}\cdots \omega_{n+|u|-1}=u\}.\]
We say that $\omega$ is {\it recurrent} if for every $u\in \Fac(\omega)$  the set $\omega\big|_u$ is infinite.
We say that $\omega$ is {\it uniformly recurrent} if for every $u\in \Fac(\omega)$  the set $\omega\big|_u$ is syndetic, i.e., of bounded gap. A word $\omega\in A^\nats$  is {\it (purely) periodic} if there exists a positive integer $p$ such that $\omega_{i+p} = \omega_i$ for all indices $i$, while it is {\it ultimately periodic} if $\omega_{i+p} = \omega_i$ for all sufficiently large $i$.  Finally, a word $\omega\in A^\nats$ is called {\it aperiodic} if it is not ultimately periodic.
For a finite word $w=w_1w_2\cdots w_n$, we call $p$ a \emph{period} of $w$ if $w_{i+p}=w_{i}$ for every $1 \le i \le n-p$. Two finite or infinite words are said to be {\it isomorphic} if the two words are equal up to a renaming of the  letters.

A (finite or infinite) word $\omega$ over $A$ is {\it balanced} if and only if for any $u,v$ factors of $\omega$ of the same length and for every letter $a\in A$, one has $||u|_{a}-|v|_{a}|\le 1$.
More generally, $\omega$ is {\it $C$-balanced} if there exists a constant $C>0$ such that  for any $u,v$ factors of $\omega$ of the same length and for every letter $a\in A$, one has $||u|_{a}-|v|_{a}|\le C$.

The \emph{factor complexity} of an infinite word $\omega$ is the function $$p_{\omega}(n)=|\Fac(\omega)\cap A^{n}|,$$ i.e., the function that counts the number of distinct factors of length $n$ of $\omega$, for every $n\geq 0$ (cf. \cite{MoHe38}).
By the Morse-Hedlund theorem, a word $\omega$ is aperiodic if and only if $p_{\omega}(n)\geq n+1$ for each $n.$ Words with $p_{\omega}(n)= n+1$ for each $n\in \nats$ are called Sturmian words.

The factor complexity counts the factors appearing in the word. A dual point of view consists in counting the shortest factors that \emph{do not} appear in the word. This leads to another measure of complexity called the {\it minimal forbidden factor complexity}.
Let $\omega$ be a (finite or infinite) word over an alphabet $A$.
A finite nonempty word $v$ is a \textit{minimal forbidden factor} for $\omega$ if $v$ does not belong to $\Fac(\omega)$ but every proper factor of $v$ does.
We let $\MF(\omega)$ denote the set of all minimal forbidden words for $\omega$.
The minimal forbidden factor complexity of an infinite word $\omega$ is the function $$mf_{\omega}(n)=\left | \MF(\omega) \cap A^{n} \right |,$$ i.e., the function that counts the number of distinct minimal forbidden factors of length $n$ of $\omega$, for every $n\geq 0$ (cf. \cite{MiReSci02}).

Another approach in measuring the complexity of a words consists in counting its factors up to an equivalence relation. The abelian complexity can be framed in this context. Two finite words $u,v$ are \emph{abelian equivalent} (denoted $u \approx v$) if they have the same Parikh vector. Note that $\approx$ is an equivalence relation over $A^{*}$.
More formally, the \emph{abelian complexity} of a word $\omega$ is the function $$a_{\omega}(n)=\left | \frac{\Fac(\omega) \cap A^{n}}{\approx} \right |,$$ i.e., the function that counts the number of distinct Parikh vectors of factors of length $n$ of $\omega$, for every $n\geq 0$ (cf. \cite{CoHe73}).

We now introduce a new measure of complexity. Recall that two finite words $u,v$ are \emph{conjugate} if there exist  words $w_{1},w_{2}$ such that $u=w_{1}w_{2}$ and $v=w_{2}w_{1}$. The conjugacy relation is an equivalence over $A^{*}$, which is denoted by $\sim$, whose classes are called \emph{conjugacy classes}. 

The \emph{cyclic complexity} of an infinite word $\omega$ is the function $$c_{\omega}(n)=\left | \frac{\Fac(\omega) \cap A^{n}}{\sim} \right |,$$ i.e., the function that counts the number of distinct conjugacy classes of factors of length $n$ of $\omega$, for every $n\geq 0$.

\begin{remark}\label{rem:ineq}
\rm{For any infinite word $\omega$ it holds that  \[a_{\omega}(n)\leq c_{\omega}(n)\leq p_{\omega}(n)\] for every $n.$ Indeed, the second inequality is obvious, while the first follows from the fact that two factors that are conjugate must have the same Parikh vector.}
\end{remark}

Another basic property of the cyclic complexity is stated in the following proposition.

\begin{proposition}\label{prop:full}
 An infinite word has full cyclic complexity if and only if it has full factor complexity.
\end{proposition}

\begin{proof}
Clearly, full factor complexity implies full cyclic complexity. Conversely, if $\omega$ is an infinite word having full cyclic complexity, then  for every  $w\in A^{*},$ some conjugate of $ww$ is an element of $\Fac(\omega).$ But as every conjugate of $ww$ contains $w$ as a factor, we have $w\in \Fac(\omega)$.
\end{proof}

Cyclic complexity, as many other mentioned complexity functions, is naturally extended to any factorial language.
Recall that a language is any subset of $A^{*}$. A language $L$ is called \emph{factorial} if it contains all the factors of its words, i.e., if  $uv\in L \Rightarrow u,v\in L$. The cyclic complexity of $L$ is defined by $$c_{L}(n)=\left | \frac{L \cap A^{n}}{\sim} \right |.$$

The cyclic complexity is an invariant for several operations on languages.
For example, it is clear that if two languages are isomorphic (i.e., one can be obtained from the other by renaming letters) then they have the same cyclic complexity. Furthermore, if $L$ is a language and $\tilde{L}$ is obtained from $L$ by reversing (mirror image) each word in $L$, then $L$ and $\tilde{L}$ have the same cyclic complexity.

\section{Cyclic Complexity Distinguishes Between Periodic and Aperiodic Words}\label{sec:theor1}

In this section we give a proof of Theorem \ref{theor1}.
The following lemma connects cyclic complexity to balancedness.

\begin{lemma}\label{lem:C}
 Let $\omega\in A^{\nats}$ and suppose that there exists a constant $C$ such that $c_{\omega}(n)\leq C$ for every $n$. Then $\omega$ is $C$-balanced.
\end{lemma}

\begin{proof}
By Remark \ref{rem:ineq}, $a_{\omega}(n)\leq C$ for every $n$. It is proved in \cite{RiSaZa11} that this implies that the word $\omega$ is $C$-balanced.
\end{proof}

\begin{lemma}\label{aper} Let $\omega\in A^\nats$ be aperiodic and let $v\in A^+$ be a factor of $\omega$ which occurs in $\omega$ an infinite number of times. Then, for each positive integer $K$ there exists a positive integer $n$ such that  $\omega$ contains at least $K+1$ distinct factors of length $n$ beginning with $v.$
\end{lemma}

\begin{proof} Let  $y_0,y_1,\ldots ,y_K$ be $K+1$ suffixes of $\omega$ beginning with $v.$ Since $\omega$ is aperiodic, the $y_i$ are pairwise distinct. Thus for all $n$ sufficiently large, the prefixes of $y_i$ of length $n$ are pairwise distinct.
\end{proof}

\noindent {\bf Theorem 1.}
{\it A word $\omega$ is ultimately periodic if and only if it has bounded cyclic complexity.
}
\medskip

\begin{proof}
If $\omega$ is ultimately periodic, then it has bounded factor complexity by the Morse-Hedlund theorem, and hence bounded cyclic complexity.

Let us now prove that if $\omega$ is aperiodic, then for any fixed positive integer $M$, $c_{\omega}(n)\geq M$ for some $n$. Short of replacing $\omega$ by a suffix of $\omega,$ we can suppose that each letter occurring in $\omega$ occurs infinitely often in $\omega.$ First, suppose that for each positive integer $C,$ $\omega$ is not $C$-balanced.  Then, by Lemma~\ref{lem:C}, the cyclic complexity of $\omega$ is unbounded and we are done. Thus, we can suppose that $\omega$ is $C$-balanced for some positive integer $C.$ 

Since $\omega$ is $C$-balanced and each $a\in \alp(\omega)$ occurs in $\omega$ an infinite number of times, it follows that there exists a positive integer $N$ such that each factor of $\omega$ of length $N$ contains an occurrence of each $a\in \alp(\omega).$
Fix $a\in \alp(\omega).$ Then $a^N$ is not a factor of $\omega.$ Let $a^k$ be the longest suffix of $a^N$ which occurs in $\omega$ an infinite number of times.  Clearly, $1\leq k < n$. So, there exists a suffix $\omega '$ of $\omega$ for which $a^{k+1}$ is a forbidden factor of $\omega'.$
By Lemma~\ref{aper}, there exists a positive integer $n_0$ such that $\omega'$ contains at least $MN$ distinct factors of length $n_0$ beginning with $a^k.$ We let  $u_1, u_2,\ldots ,u_{MN}$ denote these factors. There exist  $v_1,v_2,\ldots ,v_{MN}$, each in $A^N$, such that $u_iv_i$ are factors of $\omega'$  (of length $n_0+N)$ for each $1\leq i\leq MN.$ Since each $v_i$ contains at least one occurrence of $a,$ it follows that there exists $n>n_0$ such that $\omega'$ contains at least $M$ distinct factors of length $n$ beginning with $a^k$ and terminating in $a.$
Since $a^{k+1}$ is a forbidden factor of $\omega',$  no two of these factors are conjugate to one another. Hence, $c_{\omega'}(n)\geq M$ and thus $c_{\omega}(n)\geq M$.
\end{proof}

\section{Cyclic Complexity Distinguishes Between Sturmian Words with Different Languages}\label{sec:St}

In this section we investigate the cyclic complexity of Sturmian words and give a proof of Theorem \ref{theor2}.
We begin by reviewing some basic properties of Sturmian words which are relevant to our proof of Theorem \ref{theor2}. See also \cite[Chap. 2]{lothaire-book:2002}. Throughout this section we fix the alphabet $A=\{0,1\}.$ An infinite word $x\in A^\nats$ is called Sturmian if it satisfies any of the following equivalent conditions:

\begin{proposition}\label{prop:sturm}
 Let $x\in A^\nats$. The following conditions are equivalent:
\begin{enumerate}
 \item $x$ has exactly $n+1$ distinct factors of each length $n;$
 \item $x$ is balanced and aperiodic;
 \item $x$ has exactly one right (resp. left) special factor for each length.
\end{enumerate}
\end{proposition}

The best known example of a Sturmian word is the Fibonacci word $F = 010010100100101001\cdots$, obtained as the fixed point of the substitution $0 \mapsto 01$, $1 \mapsto 0$.
It is easy to see the set of factors of  a Sturmian word $x$ is {\em closed under reversal}, i.e., if $u$ is a factor of $x$, then so is its reversal $\tilde{u}$ (see, for instance, \cite[Chap. 2]{lothaire-book:2002}).
It follows that the right special factors of a Sturmian word are the reversals of its left special factors. In particular, the bispecial factors of a Sturmian word are palindromes.

\begin{remark}\label{why}
\rm{It follows from Proposition~\ref{prop:sturm} that if $x$ is a Sturmian word, then for each $n\geq 0$  there exist a unique factor $u$ of length $n$ such that both $u0$ and $u1$ belong to $\Fac(x)$  and a unique factor $v$ of length $n$ such that both $0v$ and $1v$ belong to  $\Fac(x).$ We consider two cases: Case 1:  $u\neq v,$ and Case 2:  $u=v.$ In Case 1 it follows that $u$ is a suffix of a unique factor $w$ of length $n+1$ and both $w0$ and $w1$ are factors of $x$ of length $n+2.$  Moreover, for each  factor $z\neq w$ of length $n+1,$ let $z'$ denote the suffix of $z$ of length $n.$ Then, as $z'$ is not right special, it follows that there exists a unique $a\in \{0,1\}$ such that  $x\big|_{z'}=x\big|_{z'a},$ that is, each occurrence of $z'$ in $x$ is an occurrence of $z'a.$ Hence $x\big|_z=x\big|_{za}.$
In other words, in Case 1 we have that  $\Fac(x)\cap \{0,1\}^{n+1}$ uniquely determines  $\Fac(x)\cap \{0,1\}^{n+2}.$
In Case 2, as $u=v$ we have that  $u$  is a bispecial factor of $x$ of length $n,$ and hence each of $u0,u1,0u,1u$ is a factor of $x$ of length $n+1.$  In this case, exactly one of the following two cases occurs: Either $0u$ is right special, in which case by the balance property we must have $x\big|_{1u}=x\big|_{1u0},$ or $1u$ is right special, in which case  $x\big|_{0u}=x\big|_{0u1}.$ Moreover, each of these two cases is possible, meaning that there exists a Sturmian word $x'$ whose factors agree with those of $x$ up to length $n+1$ and differ at length $n+2:$ One admits the factor $0u0$, while the other admits $1u1.$}
\end{remark}

The \emph{slope} of a finite nonempty word $w$ over the alphabet $A$ is defined as $s(w)=\frac{|w|_1}{|w|}$. The slope of an infinite word over $A$, when it exists, is the limit of the slopes of its prefixes.
The set of factors of a Sturmian word depends only on the  slope:

\begin{proposition}[\cite{MoHe38}]\label{MoHe}
 Let $x,y$ be two Sturmian words. Then $\Fac(x)=\Fac(y)$ if and only if $x$ and $y$ have the same slope.
\end{proposition}

Central words play a  fundamental role in the study of  Sturmian words. A word over the alphabet $A$ is {\it central} if it has relatively prime periods $p$ and $q$ and length $p+q-2$. We make use of  the following characterizations of central words (see \cite{Be07} for a survey):

\begin{proposition}\label{prop:centr}
 Let $w$ be a word over $A$. The following conditions are equivalent:

\begin{enumerate}
 \item $w$ is a central word;
 \item $0w1$ and $1w0$ are conjugate;
 \item $w$ is a bispecial factor of some Sturmian word;
 \item $w$ is a palindrome and the words $w0$ and $w1$ (resp.~$0w$ and $1w$) are balanced;
 \item $0w1$ is balanced and is the least element (relative to the lexicographic order) in its conjugacy class;
 \item $w$ is a power of a single letter or there exist central words $p_{1},p_{2}$ such that $w=p_{1}01p_{2}=p_{2}10p_{1}$. Moreover, in this latter case $|p_{1}|+2$ and $|p_{2}|+2$ are relatively prime periods of $w$ and $\min(|p_{1}|+2,|p_{2}|+2)$ is the minimal period of $w$.
\end{enumerate}
\end{proposition}

Let $w$ be a central word, different from a power of a single letter, having relatively prime periods $p$ and $q$ and length $p+q-2$. The words $0w1$ and $1w0$, which, by Proposition \ref{prop:centr}, are conjugate, are called \emph{Christoffel words}. Let $r=|0w1|_{0}$ and $s=|0w1|_{1}$. It can be proved that $\{r,s\}=\{p^{-1},q^{-1}\}$ modulo $p+q$ \cite[Proposition 2.1]{BeDelRe08}. Moreover, the conjugacy class of $0w1$ and $1w0$ contains exactly $|w|+2$ words. If we sort these words lexicographically and arrange them as rows of a matrix, we obtain a square matrix with remarkable combinatorial properties (see \cite{BoRe06,JeZa04,MaReSc03}).
This matrix depends only on the pair $(r,s)$; we call it
the \emph{$(r,s)$-Christoffel array} and denote it by $\mathcal{A}_{r,s}.$ Two consecutive rows of $\mathcal{A}_{r,s}$ differ only by a swap of two consecutive positions \cite[Corollary 5.1]{BoRe06}. Moreover, the columns are also conjugate and in particular the first one is $0^{r}1^{s}$, while the last one is $1^{s}0^{r}$ (cf.~\cite{MaReSc03}).
For example, consider the Fibonacci word $F$ and its bispecial factor $w=010010$, which has periods $p=5$ and $q=3$. We have $s=q^{-1}=3<5=r=p^{-1}$. In Figure \ref{fig:matrix} we show the $(5,3)$-Christoffel array $\mathcal{A}_{5,3}$. The rows are the lexicographically sorted factors of $F$ with Parikh vector $(5,3)$. The other factor of length $8$ of $F$ is $10100101$.

\begin{figure}
\begin{center}
$$
\mathcal{A}_{5,3}=
\left(
{\begin{tabular}{*{8}{@{\:\:}c}@{}}
 0 & 0 & 1 & 0 & 0 & 1 & 0 & 1\\
 0 & 0 & 1 & 0 & 1 & 0 & 0 & 1\\
 0 & 1 & 0 & 0 & 1 & 0 & 0 & 1\\
 0 & 1 & 0 & 0 & 1 & 0 & 1 & 0\\
 0 & 1 & 0 & 1 & 0 & 0 & 1 & 0\\
 1 & 0 & 0 & 1 & 0 & 0 & 1 & 0\\
 1 & 0 & 0 & 1 & 0 & 1 & 0 & 0\\
 1 & 0 & 1 & 0 & 0 & 1 & 0 & 0
\end{tabular}}
\hspace{1.5mm}\right) $$
\end{center}
\caption{The Christoffel array $ \mathcal{A}_{5,3}$.
\label{fig:matrix}}
\end{figure}

Every aperiodic word (and therefore, in particular, every Sturmian word) contains infinitely many bispecial factors. If $w$ is a bispecial factor of a Sturmian word $x$, then $w$ is central by Proposition \ref{prop:centr}. Moreover, there exists a unique letter $a \in A$ such $aw$ is right special, or equivalently $wa$ is left special. Also, the next (by length) bispecial factor
$w'$ of $x$ is the shortest palindrome beginning with $wa.$ If $p$ and $q$ are the relatively prime periods of $w$ such that $|w|=p+q-2$, then the word $w'$ is central with relatively prime periods $p'$ and $q'$ verifying $|w'|=p'+q'-2$ and either $p'=p+q$ and $q'=p$, or $p'=p+q$ and $q'=q$, depending on the letter $a$. For example, $010$ is a bispecial factor of the Fibonacci word $F$ and has relatively prime periods $3$ and $2$ (and length $3+2-2$). The successive (in length order) bispecial factor of $F$ is $010010$, which is the shortest palindrome beginning with $010\cdot 0$ and has relatively prime  periods $5$ and $3$ (and length $5+3-2$). There exist other Sturmian words having $010$ as a bispecial factor and for which the successive bispecial factor is $01010$ (i.e., the shortest palindrome beginning with $010\cdot 1$) whose relatively prime periods are $5$ and $2$.
These combinatorial properties of central words are needed in our proof of Theorem \ref{theor2}.

\medskip

While Sturmian words have unbounded cyclic complexity (see Theorem \ref{theor1}),  their cyclic complexity takes value $2$ infinitely often.  More precisely, we have the following result.

\begin{lemma}\label{lem:bis+2}
 Let $x$ be a Sturmian word. Then $c_{x}(n)=2$ if and only if $n=1$ or there exists a bispecial factor of $x$ of length $n-2$. Moreover, when $c_{x}(n)=2$, one conjugacy class has cardinality $n$ and the other has cardinality $1$.
\end{lemma}

\begin{proof}
If $n=1$ clearly $c_{x}(n)=2$, so let us suppose $n>1$. If $w$ is a bispecial factor of length $n-2$ of a Sturmian word, then there exists a letter $a$ such that the factors of $x$ of length $n$ are precisely $awa$, and all the conjugates of $awb$ for the letter $b\neq a$  (cf.~\cite{DelMig94}). Hence $c_{x}(n)=2$.
Conversely, suppose $c_{x}(n)=2$ for some $n>1$. Then the cyclic classes of factors of $x$ of length $n$ correspond to the abelian classes of factors of $x$ of length $n$. Let $w$ (possibly empty) be the left special factor of $x$ of length $n-2.$
Then $01w$ and $10w$ are factors of $x.$  Since  $01w$ and $10w$ are abelian equivalent, they must be conjugate. Thus by Proposition \ref{prop:centr}, $w$ is a palindrome, whence $w$ is a bispecial factor of $x$.
\end{proof}

It follows that  $\liminf_{n\to \infty} c_{x}(n)=2$. However, as we see in Section \ref{sec:tm}, this is not a characterization of Sturmian words.

\medskip
\noindent {\bf Theorem 2.}
{\it Let $x$ be a Sturmian word. If a word $y$ has the same cyclic complexity
as $x$ then, up to renaming letters,  $y$ is a Sturmian word having the same slope as $x$.
}
\medskip

\begin{proof}
Since $y$ has the same cyclic complexity as $x$, we have that in particular $2=c_{x}(1)=c_{y}(1)$, so $y$ is a binary word. We fix for $x$ and $y$ the alphabet $\{0,1\}$. Since $x$ is aperiodic,  by Theorem \ref{theor1} $c_{x}$ is unbounded. Since $x$ and $y$ have the same cyclic complexity  we have, still by Theorem \ref{theor1}, that $y$ is aperiodic.

Up to exchanging $0$ and $1$ in $x,$  we can assume that $x$ contains the factor $00$ so that the factors of $x$  of length $2$ are $00,01,10.$  We claim that $y$ too has exactly three factors of length $2.$ In fact, since $y$ is aperiodic, $y$ has at least three distinct factors of length $2.$ If $y$ had four factors of length $2$, then $y$ would have three abelian classes of length $2$ and hence $c_y(2)=3,$  a contradiction. Thus, up to exchanging $0$ and $1$ in $y$, we can assume that $x$ and $y$ have the same factors of length $2$.

We now prove that $x$ and $y$ have the same set of factors. This implies that $y$ is a Sturmian word and has the same slope as $x$ by Proposition \ref{MoHe}.
Suppose to the contrary that there exists a least positive integer $n>0$ such that the factors of $x$ and $y$ of length $n+2$ differ.
In what follows, we assume:

(*) Let $x,y\in \{0,1\}^\nats$ be infinite aperiodic words having the same factors of length $2.$ Assume further that $x$ is Sturmian and $\Fac(x)\neq \Fac(y).$ Let $n$ be the least positive integer such that the factors of $x$ and $y$ of length $n+2$ differ. Let $a\in \{0,1\}$ and $w\in \{0,1\}^n$ be such that $aw$ is the unique right special factor of $x$ of length $n+1.$ Let $b=1-a$ so that $\{a,b\}=\{0,1\}.$

\medskip

\begin{lemma} Assume $x$ and $y$ satisfy (*). Then $x$ and $y$ have a common bispecial factor of length $n.$
\end{lemma}

\begin{proof} This is essentially Case 2 in Remark~\ref{why}. We begin by observing that $w$ is the unique right special factor of $x$ and of $y$ of length $n.$
We claim that $w$ is a bispecial factor of both $x$ and $y.$ If not, then $aw$ is the unique right special factor of both $x$ and $y$ of length $n+1,$ and so $x$ and $y$ would have the same set of factors of length $n+2,$ a contradiction. \end{proof}

\begin{lemma} Assume $x$ and $y$ satisfy (*). Then $c_x(n+2)=2.$
\end{lemma}

\begin{proof} This follows immediately from the previous lemma together with Lemma \ref{lem:bis+2}. \end{proof}

\begin{lemma}\label{lem:bw}
Assume $x$ and $y$ satisfy (*). Then either $c_y\neq c_x,$ or  $bw$ is the unique right special factor of $y$ of length $n+1$ and every occurrence of $aw$ in $y$ is followed by $b$.\end{lemma}

\begin{proof}Assume  $c_y=c_x.$ Then by the previous lemma we have $c_x(n+2)=c_y(n+2)=2.$ It follows that exactly one of  $aw$ or $bw$ is right special in $y.$ Since $y$ is aperiodic, at least one of the two must be right special.  On the other hand, if both were right special, then  $y$ would have at least $3$ abelian classes of factors of length $n+2$ (namely those of $awa,bwb$ and $awb)$ whence  $c_y(n+2)\geq 3,$ a contradiction. We claim that $bw$ is right special in $y.$ In fact,
suppose to the contrary that $aw$ is right special in $y$. In this case, exactly one of $bwb$ and $bwa$ is a factor of $y$. If $bwb$ is a factor of $y$, then as above $y$ would have at least three abelian classes of factors of length $n+2$ and hence $c_y(n+2)\geq 3,$ a contradiction. If $bwa$ is a factor of $y$, then $x$ and $y$ have the same factors of length $n+2$, a contradiction. This proves that $bw$ is the unique right special factor of $y$ of length $n+1.$ As before, exactly one of $awa$ and $awb$ is a factor of $y$. If $awa$ is a factor of $y$, then as argued above $y$ would have at least three abelian classes of factors of length $n+2$ and hence $c_y(n+2)\geq 3,$ a contradiction. Thus  $y\big|_{aw}=y\big|_{awb}.$ \end{proof}

By Remark~\ref{why}, there exists a Sturmian word $y'$ such that $x$ and $y'$ have the same set of factors up to length $n+1,$ after which $aw$ is right special in $x$ while $bw$ is right special in $y'.$
Thus by the previous lemma, if $c_x=c_y,$ then  $y$ and $y'$ have the same factors  of length $n+2.$ Let  $w_{x}$ (resp.~$w_{y}$) be the shortest bispecial factor of $x$ (resp.~of $y$) whose length is strictly greater than $|w|=n$.
Note that $w_{y}$ is also bispecial for $y'$.
Then $\Fac(y')\cap A^j=\Fac(y)\cap A^j$ for every $j\leq |w_y|+1$.

\begin{lemma} Assume $x$ and $y$ satisfy (*). Then either $c_y\neq c_x,$ or  $|w_{x}|>|w_{y}|.$\end{lemma}

\begin{proof} Assume $c_y=c_x.$ Then as in the previous lemma we have $c_x(n+2)=c_y(n+2)=2.$ Let $p'$ and $q'$, with $p'>q'$, be the two relatively prime periods of $w$ such that $n=|w|=p'+q'-2$.
We have that $\{|w_{x}|,|w_{y}|\}=\{2p'+q'-2,p'+2q'-2\}$, hence $w_{x}$ and $w_{y}$ cannot have the same length.
If $|w_{x}|<|w_{y}|$, then by Lemma \ref{lem:bis+2}
\[
2=c_x(|w_x|+2)=c_y(|w_x|+2)=c_{y'}(|w_x|+2)>2,
\]
a contradiction. \end{proof}

Let $p'$ and $q'$, with $p'>q'$, be the two relatively prime periods of $w.$  By the previous lemma we have $|w_{x}|>|w_{y}|.$  So $w_{x}$ has periods $p'+q'$ and $p'$ and length $2p'+q'-2$, while $w_{y}$ has periods $p'+q'$ and $q'$ and length $p'+2q'-2.$
Set $p=p'+q'$ and $q=p'$, so that $|w_{y}|=2p-q-2$ and $|w_{x}|=p+q-2$.  Notice that $p+q>2p-q$ since $p'>q'$. We use this fact in what follows without explicit mention.

\begin{lemma} Assume $x$ and $y$ satisfy (*). Then either $c_y\neq c_x,$ or $w_{y}$ is a strongly bispecial factor of $y$, i.e., $0w_{y}0$, $0w_{y}1$, $1w_{y}0$ and $1w_{y}1$ are all factors of $y$.
\end{lemma}

\begin{proof} Assume that $c_y=c_x.$ Then one of the following cases must hold: 

\begin{enumerate}
\item Neither $0w_y$ nor $1w_y$ is right special in $y$;
\item $0w_y$ is right special in $y$ and every occurrence of $1w_y$ is followed by $1$ in $y$;
\item $0w_y$ is right special in $y$ and every occurrence of $1w_y$ is followed by $0$ in $y$;
\item $1w_y$ is right special in $y$ and every occurrence of $0w_y$ is followed by $0$ in $y$;
\item $1w_y$ is right special in $y$ and every occurrence of $0w_y$ is followed by $1$ in $y$;
\item Both $0w_y$ and $1w_y$ are right special factors of $y$.
\end{enumerate}

In Case 1 $y$ does not have right special factors of length $|w_y|+1$, hence $y$ would be ultimately periodic, a contradiction.

Case 2 also implies that $y$ is ultimately periodic. If no nonempty prefix of $1w_y$ is right special in $y$, then every occurrence of $1$ in $y$ is an occurrence of $1w_y1$, and hence $y$ is ultimately periodic. Let $z$ (possibly empty) be the longest prefix of $w_y$ such that $1z$ is right special in $y$. Clearly, $|z|<|w_y|$. So we can write $1w_y=1zu$ for some nonempty word $u$. Since $1z$ is right special in $y$ and hence in $y'$, we have that $\tilde{z}1$ is left special in $y'$ and hence in $y$. Thus $\tilde{z}1$ is a prefix of $w_y$, whence $u$ begins with $1$ and $z=\tilde{z}$. Therefore each occurrence of $1z1$ is an occurrence of $1zu1=1w_y1$. Since $1w_y1$ and $1z1$  are both palindromes, $1z1$ is also a suffix of $1w_y1$, whence $y$ is ultimately periodic.

In Case 3, either $\Fac(y')\cap A^j=\Fac(y)\cap A^j$ for every $j\leq |w_y|+2$, and hence $2=c_{y'}(|w_y|+2)=c_{y}(|w_y|+2)=c_{x}(|w_y|+2)>2$, contradiction, or by Remark~\ref{why}
there exists a Sturmian word $y''$ such that $\Fac(y'')\cap A^j=\Fac(y)\cap A^j$ for every $j\leq |w_y|+2,$ in which case $2=c_{y''}(|w_y|+2)=c_{y}(|w_y|+2)=c_{x}(|w_y|+2)>2$, again a contradiction.

Case 4 is symmetric to Case 2 and Case 5 is symmetric to Case 3, so the only remaining case is that both $0w_y$ and $1w_y$ are right special factors of $y$ as required. \end{proof}

\begin{lemma}Assume $x$ and $y$ satisfy (*). Then either $c_y\neq c_x,$ or $c_y(|w_y|+2)=c_y(2p-q)=3.$
\end{lemma}

\begin{proof}Assume $c_x=c_y.$ Then from the previous lemma we have that $w_y$ is a strong bispecial factor of $y.$ Thus,   the factors of $y$ of length $|w_y|+2$ are precisely the factors of $y'$ of the same length, plus one other factor which is either $0w_y0$ or $1w_y1$, whence $c_y(|w_y|+2)=c_y(2p-q)=3$.
\end{proof}

Returning to the proof of Theorem \ref{theor2}, let $r=|0w_{x}1|_{0}$ and $s=|0w_{x}1|_{1}$. Since we supposed that $11$ is not a factor of $x$, we have $r>s$.  In what follows we assume that $c_x=c_y.$
In view of the previous lemmas, we have that if $x$ and $y$ satisfy (*),  then $|w_x|>|w_y|$ and $c_y(|w_y|+2)=c_y(2p-q)=3.$ We consider four cases depending on $s:$  $s=1,$ $s=2,$ $s=3$ and $s>3.$ Each gives rise to a contradiction.

\medskip

\emph{Case $s=1$}. This case cannot happen since otherwise we would have $w_{x}=0^{n+1}$, $w=0^{n}$ and $w_{y}=0^{n}10^{n}$,  against the hypothesis that $|w_{x}|>|w_{y}|$.

\medskip

\emph{Case $s=2$}. In this case we have $w=0^{n}$, $w_{x}=0^{n}10^{n}$ and $w_{y}=0^{n+1}$. Since $w_{x}$ is right special in $x$, we have that $10^{n}1$ and $0^{n+1}$ are factors of $x$. Let us look at the factors of $x$ of length $2n+4$. Among them, we have $v_{1}=10^{n}10^{n+1}1$, $v_{2}=10^{n+1}10^{n}1$ and $v_{3}=0^{j}10^{n}10^{k}$, for some $j,k$ such that $j+k=n+2$. Moreover, since one of $10^{n}10^{n}1$ and $10^{n+1}10^{n+1}$ is a factor of $x$, we also have that either $v_{4}=10^{n}10^{n}10$ or $v'_{4}=10^{n+1}10^{n+1}$ is a factor of $x$. Since these four factors are not conjugate to one other, we have $c_{x}(2n+4)\geq 4$.

Let us now prove that $c_{y}(2n+4)=3$. Since $w_y=0^{n+1}$ is strongly bispecial in $y$, we have that $0^{n+2}1$ is a factor of $y$ and hence there exists a factor of $y$ of length $2n+3$ beginning with $0^{n+2}1$. This factor must be equal to $0^{n+2}10^{n}$ since otherwise $y$ would contain both $0^{t+2}$ and $10^{t}1$ for some $t\leq n-1$, and hence also $x$ would contain these factors, against the hypothesis that $x$ is Sturmian and therefore balanced. We have thus proved that $y$ contains a factor of length $2n+3$ with exactly one $1$. Since $2n+3$ is the length of a bispecial factor of $x$ plus 2, we have by Lemma \ref{lem:bis+2} that $c_{y}(2n+3)=c_{x}(2n+3)=2$. Since $y$ contains factors of length $2n+3$ with two $1$'s, these must be all conjugates one to each other. Since $w_{y}$ is a strongly bispecial factor of $y$, we have that $10^{n+1}1$ is a factor of $y$ and therefore the factors of length $2n+3$ of $y$ with two $1$'s are all conjugate to $10^{n+1}10^{n}$.

By Lemma \ref{lem:bw}, $10^n1$ is not a factor of $y$; neither is $10^t1$ for $t<n$.
Let $10^{t}1$ be a factor of $y$, with $t>n+1$. Then we can prove that $t=2n+2$. Indeed, on the one hand we cannot have $t>2n+2$ since $0^{2n+3}$ is not a factor of $y$---because we know that the factors of length $2n+3$ of $y$ contain one or two $1$'s. On the other hand, we cannot have $t<2n+2$ because we know that the factors of length $2n+3$ with exactly two $1$'s are all conjugate to $10^{n+1}10^{n}$. We have therefore proved that in $y$ two consecutive occurrences of $1$ are separated by either $n+1$ or $2n+2$ many $0$'s. This implies that at length $2n+4$ we have in $y$: one conjugacy class containing all the factors with exactly one $1$; one conjugacy class containing only the factor $10^{2n+2}1$; one conjugacy class containing all the other factors, that are of the form $0^{j}10^{n+1}10^{k}$, with $j+k=n+1$. Hence, $c_{y}(2n+4)=3$ and we are done.

\medskip

\emph{Case $s=3$}. In this case $w_x$ is a central word with two $1$'s. The only central words with two $1$'s not containing $11$ as a factor are of the form $0^m10^m10^m$ or $0^m10^{m+1}10^m$ for some $m>0$. This implies that $w=0^m10^m$, and since $|w_x|>|w_y|$, we deduce  $w_y=0^m10^m10^m$ and $w_x=0^m10^{m+1}10^m$. Note that we have $m=p'-2=q'-1$, so that $2p-q=3m+4$.

It is readily verified that each of $0^{m+1}10^{m+1}10^{m}$, $0^{m+1}10^{m}10^{m+1}$, $0^{m}10^{m+1}10^{m}1$ and $0^{m-1}10^{m+1}10^{m+1}1$ is a factor of $x$ and no two of them are conjugate. Therefore, $c_{x}(2p-q)\geq 4$, contradicting that $c_{y}(2p-q)=3$.

\medskip

\emph{Case $s>3$}.  As in the previous case, we prove that $c_{x}(2p-q)\geq 4$. It is known that among the $p+q+1$ factors of $x$ of length $p+q$, there is one factor with a Parikh vector $\mathcal{Q}$ and the remaining $p+q$ factors with the other Parikh vector $\mathcal{Q'}$, these latter being in the same conjugacy class, which is in fact the conjugacy class of the Christoffel word $0w_{x}1$ (see Lemma \ref{lem:bis+2}).

We can build the $(r,s)$-Christoffel array $\mathcal{A}_{r,s}$ (recall that $r+s=p+q$).
The factors of length $2p-q$ of $x$ can be obtained by removing the last $2q-p$ columns from $\mathcal{A}_{r,s}$ (of course, in this way some rows can be equal and therefore some factors appear more than once). Let $ \mathcal{A'}_{r,s}$ be the matrix made up of the first $2p-q$ columns of $\mathcal{A}_{r,s}.$  In what follows, we let $\mathcal{A'}_i$ denote the $i$-th row of $\mathcal{A'}_{r,s}$. Recall that $\{r,s\}=\{p^{-1},q^{-1}\}\mod (p+q)$. We separate  two subcases: $s=p^{-1}$ or $s=q^{-1}.$

 \begin{remark}\label{rem:abelian-classes}
\rm{Before treating the two remaining subcases in the proof of Theorem~\ref{theor2}, we recall here some properties of the arrays $\mathcal{A}_{r,s}$ and $\mathcal{A'}_{r,s}$ which will be used. Each column of $\mathcal{A}_{r,s}$ and $\mathcal{A'}_{r,s}$ has $r+s$ entries. The first column in each case has $r$-many $0$'s at the top followed by $s$-many $1$'s at the bottom.
Then each subsequent column is obtained from the previous column by rotating upwards by an amount equal to $s$ as illustrated in  Figure~\ref{fig:matrix}. Any two consecutive rows of the array  $\mathcal{A}_{r,s}$ differ precisely in two consecutive positions where the upper row has $01$ and the lower row $10.$ Thus two consecutive rows $\mathcal{A'}_i$ and $\mathcal{A'}_{i+1}$ of  $ \mathcal{A'}_{r,s}$ either differ in the same way in  two consecutive positions,  or are equal, or differ only in their last entry. They are therefore abelian equivalent, except in the last case, in which case all rows $\mathcal{A'}_j$ for $j\leq i$ are abelian equivalent and all rows $\mathcal{A'}_j$ with $j\geq i+1$ are abelian equivalent.}
\end{remark}

\medskip

\emph{Subcase $s=p^{-1}$}. In this case, we prove that the bottom three rows in $ \mathcal{A'}_{r,s}$ are distinct and begin and end with $1$. It follows that each of these rows is unique in its conjugacy class since all other conjugates contain an occurrence of $11.$ Together with the first row of $ \mathcal{A'}_{r,s}$, which is not abelian equivalent to any of the bottom three rows, we obtain at least four conjugacy classes of factors of $x$ of length $2p-q$. This subcase is depicted in Figure \ref{fig:matrix_rs}.

Since $s\geq 3,$ it follows that the bottom three rows in $\mathcal{A'}_{r,s}$ begin with $1.$ Because $sp=1\mod(p+q)$, writing $2p-q=3p-(p+q)$ it follows that $s(2p-q)=3\mod(p+q)$ which means that the bottom three rows each end with $1$ and are pairwise abelian equivalent. 
Moreover, the bottom three rows in $\mathcal{A'}_{r,s}$ are distinct. In fact, since $sp=1\mod(p+q),$ it follows that the bottom two rows differ in the $p$'th entry. More precisely, the $p$'th column of $\mathcal{A'}_{r,s}$ has $(s-1)$-many $1$'s at the top, followed by $r$-many $0$'s then a single $1$ at the bottom.
Similarly, because $s(p-q)=2\mod(p+q),$ it follows that  $\mathcal{A'}_{p+q-2}$ differs from each of $\mathcal{A'}_{p+q-1}$ and $\mathcal{A'}_{p+q}$ in the $(p-q)$'th entry ($\mathcal{A'}_{p+q-2}$ has a $0$ while the other two have a $1).$

\begin{figure}
$$\mathcal{A'}_{r,s}=
\bordermatrix{
  &1       &        & p-q       &        & p      &        &2p-q     \cr
  &\vdots  &\ddots  & \vdots    &\ddots  & \vdots & \ddots &\vdots   \cr
  &\vdots  &\ddots  & \vdots    &\ddots  & \vdots & \ddots &\vdots   \cr
   &1      &\cdots  & \cdots    &\cdots  & \cdots & \cdots & 0   \cr
   &1      & \cdots & 0         & \cdots & \cdots & \cdots & 1      \cr
   &1      & \cdots & 1         & \cdots & 0      & \cdots & 1      \cr
   &1      & \cdots & 1         & \cdots & 1      & \cdots & 1      \cr
}$$

\caption{The matrix $ \mathcal{A'}_{r,s}$ in the Subcase $s=p^{-1}$ in the proof of Theorem \ref{theor2}.\label{fig:matrix_rs}}
\end{figure}
%

\medskip

\begin{figure}
$$\mathcal{A'}_{r,s}=
\bordermatrix{
  &1       &        & p-q       &        & p      &        &2p-q   \cr
   &0      &v  & 0    &1y  & 0 & 1z & 0   \cr
   &0      & v & 0    & 1y & 1 & 0z & 0      \cr
   &0      & v & 1    & 0y & 1 & 0z & 0      \cr
   &\vdots      & \cdots & 1         & \cdots & 1      & \cdots & 1      \cr
   &\vdots  &\ddots  & \vdots    &\ddots  & \vdots & \ddots &\vdots   \cr
  &\vdots  &\ddots  & \vdots    &\ddots  & \vdots & \ddots &\vdots   \cr
}$$

\caption{The matrix $ \mathcal{A'}_{r,s}$ in the Subcase $s=q^{-1}$ in the proof of Theorem \ref{theor2}.\label{fig:lyon}}
\end{figure}

\emph{Subcase $s=q^{-1}$}. In this case we prove that  the top three rows of the matrix $ \mathcal{A'}_{r,s}$ are pairwise distinct, neither is conjugate to another, and are pairwise abelian equivalent. Combined with the bottom row, which is not abelian equivalent to the top three rows, we obtain at least four distinct conjugacy classes of $x$ of length $2p-q$.

Reasoning as in the previous subcase, since $sp=-1\mod(p+q),$ it follows that the first row differs from the second and third rows in the $p$'th entry (the first row has a $0$ while the other two have $1).$ Thus  $\mathcal{A'}_1\neq \mathcal{A'}_2$ and $\mathcal{A'}_1\neq \mathcal{A'}_3.$ Since $s(p-q)=-2\mod(p+q)$
it follows that $\mathcal{A'}_2$ differs from $ \mathcal{A'}_3$ in the $(p-q)$'th entry.
And since $s(2p-q)=-3\mod(p+q)$, it follows that the top three rows
are in a different abelian class than all other rows.

Clearly  $\mathcal{A'}_1$ begins with the maximum run of $0$'s. Since $\mathcal{A'}_1$ and $\mathcal{A'}_2$ differ only in positions $p$ and $p+1$ and since both have a $1$ in position $p-q+1,$ it follows that $\mathcal{A'}_2$ also begins with the maximum run of $0$'s. Since $s(2p-q)=-3\mod(p+q)$, it follows that $\mathcal{A'}_3$ differs from $\mathcal{A'}_4$ only in the last entry, and hence $\mathcal{A'}_3$ ends with the maximum run of $0$'s. Since $\mathcal{A'}_2$ and $\mathcal{A'}_3$ differ only in positions $p-q$ and $p-q+1$, it follows that $\mathcal{A'}_2$ too ends with the maximum run of $0$'s. Hence, $\mathcal{A'}_2$ is unique in its conjugacy class.

It remains to show that $\mathcal{A'}_1$ and $\mathcal{A'}_3$ are not conjugate.
We can write $\mathcal{A'}_1=0v01y01z0$ and $\mathcal{A'}_3=0v10y10z0$ for words $v,y,z$ such that $|0v0|=p-q$ and $|0v01y0|=p$ (see Figure \ref{fig:lyon}). Let $m$ be such that $0^{m+1}1$ is a prefix of $\mathcal{A'}_1$. Then $0^{m}$ is a suffix of $\mathcal{A'}_1$ since the top row of  $ \mathcal{A}_{r,s}$ has a $1$ in column $2p-q+1.$
Thus the only conjugate of $\mathcal{A'}_1$ that is possibly a factor of $x$ is $v01y01z00$, as all other conjugates contain the forbidden factor $0^{m+2}$.
If $\mathcal{A'}_1$ and $\mathcal{A'}_3$ were conjugate, we would therefore have $v01y01z00=0v10y10z0$, hence each of $v,y,z$ would be a power of $0$. Since $w$ is the prefix of $0^{-1}\mathcal{A'}_1$ of length $p-2$, we deduce that $|w|_1\leq 1$, whence  $|w_x|_1\leq 2$, and thus $s\leq 3$, contradicting our assumption that $s>3.$ This completes the proof of Theorem \ref{theor2}.\end{proof}

Having established Theorem \ref{theor2}, one may ask: Given two infinite words $x$ and $y$ with the same cyclic complexity, what can be said about their languages of factors?
First, there exist two periodic words having same cyclic complexity but whose languages of factors are not isomorphic nor related by mirror image. For example, let $\tau$ be the morphism: $0\mapsto 010$, $1\mapsto 011$ and consider the words $x=\tau((010011)^\omega)$ and $x'=\tau((101100)^\omega)$.
One can verify that $x$ and $x'$ have the same cyclic complexity up to length $17$ and, since
each has period $18$, the cyclic complexities of $x$ and $y$ agree for all $n.$ Furthermore, it is easy to show that even two aperiodic words can have the same cyclic complexity but different languages of factors. For example, let $x$ be an infinite binary word such that $\MF(x)=\{000111\}$ and $y$ an infinite binary word such that $\MF(y)=\{001111\}$. Then the languages of factors of $x$ and $y$ are not isomorphic, nor related by mirror image, yet the two words have the same cyclic complexity. However, we do not know if this can still happen with the additional hypothesis of linear complexity, for example.

We end this section by comparing cyclic complexity and minimal forbidden factor complexity. In \cite{MiReSci02} the authors proved the following result.

\begin{theorem}\label{theor:mff}
 Let $x$ be a Sturmian word and let $y$ be an infinite word  such that for every $n$ one has $p_{x}(n)=p_{y}(n)$ and $mf_{x}(n)=mf_{y}(n)$, i.e., $y$ is a word having the same factor complexity and the same minimal forbidden factor complexity as $x$. Then, up to isomorphism, $y$ is a Sturmian word having the same slope as $x$.
\end{theorem}

In contrast with Theorem \ref{theor2}, the fact that $y$ is a Sturmian word  in  Theorem \ref{theor:mff} follows immediately from the hypothesis that $y$ has the same factor complexity as $x$.
Let $x$ be an infinite binary word  such that
$\MF(x)=\{11,000\}$ and $y$ an infinite binary word such that $\MF(y)=\{11,101\}$. Then $x$ and $y$ have the same minimal forbidden factor complexity, but it is readily checked that $c_x(5)=3$ while $c_y(5)=4$. Note that $x$ contains $7$ factors of length $5$ corresponding to $3$ cyclic classes $(00100, 00101, 01001, 01010, 10010, 10100, 10101)$  while $y$ contains the factors
$00000, 10000, 10010, 10001$ no two of which are cyclically conjugate.

\section{The Limit Inferior of the Cyclic Complexity}\label{sec:tm}

We say that an aperiodic word $x$ has  \emph{minimal cyclic complexity} if $\liminf_{n\to \infty} c_x(n)=2.$
In the previous section we proved that Sturmian words have minimal cyclic complexity.  We now give other  examples of words having minimal cyclic complexity which include the well-known period-doubling word.  This may be compared with an analogous situation in the context of maximal pattern complexity in which a restricted class of Toeplitz words is found to have the same maximal pattern complexity as Sturmian words (see \cite{KaZa01}). We also show that for the paperfolding word we have $\liminf_{n\to \infty} c_x(n)=4.$
Clearly, if $\liminf_{n\to \infty} c_x(n)<+\infty,$ then the factor complexity of $x$ satisfies $\liminf_{n\to \infty} p_x(n)/n<+\infty$.
This is because each cyclic class of factors of length $n$ has at most $n$ elements.  But the converse is not true. In fact, we prove that for the Thue-Morse infinite word $t,$ for which $\liminf_{n\to \infty} p_t(n)/n=3$, see \cite[Proposition 4.5]{Br89},
we have  $\liminf_{n\to \infty} c_t(n)=+\infty.$

\subsection{Fixed Points of Uniform Substitutions with one Discrepancy}

\begin{proposition}\label{pro:unif}
 Let $A=\{0,1\}$ and $\mu:0\mapsto u0v,1\mapsto u1v$, for words $u,v\in A^*$ such that $|uv|>0$. Let $x$ be a fixed point of $\mu$. If $x$ is aperiodic, then for every $n\geq 0$ one has $c_x(k^n)=2$, where $k=|u|+|v|+1$.
\end{proposition}

\begin{proof}
We proceed by induction on $n$. Since $x$ is binary, the result is trivially verified for $n=0.$ Next let us  fix $n\geq 1,$ and we suppose by induction hypothesis that the result is true up to $n,$ and prove it for $n+1.$ We separate the factors of $x$ of length $k^{n+1}$ into two classes: those which are images under $\mu$ of a factor of length $k^n$, and those which are not. Clearly if two words of length $k^n$ are conjugate, then so are  their images under $\mu.$ Whence if we restrict to factors of length $k^{n+1}$ which are images under $\mu$ of factors of length $k^n,$ then there are  at most $c_x(k^n)$ many cyclic classes.  Next we show that each factor of length $k^{n+1}$ which is not the image under $\mu$ of a factor of length $k^n$ is conjugate to one that is.  So let $z$ be a factor of $x$ of length $k^{n+1}$ which is not the image under $\mu$ of a factor of $x$ of length $k^n.$ Then either $z=u'a_1vua_2v\cdots ua_{n}vu''$ for letters $a_i\in A$ and words $u',u''$ such that $u''u'=u$ or $z=v''ua_1vua_2v\cdots ua_{n}v'$ for letters $a_i\in A$ and words $v',v''$ such that $v'v''=v$. In either case $z$ is conjugate of $z'=ua_1vua_2v\cdots ua_{n}v$, which is an image under $\mu$ of a factor of $x$ of length $k^n$. Thus the number of cyclic classes of factors of length $k^{n+1}$ is at most $c_x(k^n)$ which by induction hypothesis is equal to $2.$ Since $x$ is aperiodic, we deduce that $c_x(k^{n+1})=2$ as required.
\end{proof}

\begin{example}
 If we take $u=0$ and $v=\epsilon$, we obtain the morphism $\mu:0\mapsto 00,1\mapsto 01$, whose fixed point is the so-called \emph{period-doubling word} $p=0100010101000100\cdots$. By Proposition \ref{pro:unif}, we have $c_p(2^n)=2$ for every $n\geq 0$.
\end{example}

More generally, we can consider words that are obtained as a limit of a sequence of substitutions each of the form $\mu_i$  defined by $\mu_i(a)= u_iav_i$ for  $a\in \{0,1\}$, where $u_i,v_i\in A^*$ are such that $|u_i|>0$. Indeed, one can define the infinite word $x=\lim_{n\to\infty} \mu_1\circ \mu_2\circ \cdots \circ \mu_n(0)$, since the words in the sequence have arbitrarily long common prefixes. By a similar argument as that used in the proof of  Proposition \ref{pro:unif}, we have that the following proposition holds.

\begin{proposition}\label{pro:mor}
Let $(\mu_i)_{i\geq 1}$ be an infinite sequence of substitutions such that for every $i$ there exist $u_i,v_i\in A^*$, $|u_i|>0$ and $\mu_i(a)=u_iav_i$ for each $a\in \{0,1\}$. Let $x=\lim_{n\to\infty} \mu_1\circ \mu_2\circ \cdots \circ \mu_n(0)$. 
If $x$ is aperiodic, then $\liminf_{n\to \infty} c_x(n)=2$.
\end{proposition}


\subsection{Paperfolding Words}

A paperfolding word is the sequence of ridges and valleys obtained
by unfolding a sheet of paper which has been folded in half infinitely many times. 
For example, the  regular paperfolding word \[x=00100110001101100010011100110110\cdots\] is obtained by folding a sheet of paper repeatedly in half in the same direction. 
Alternatively, an infinite word $x=x_0x_1x_2\cdots \in \{0,1\}^\nats$ is a paperfolding word if $(x_{4n})_{n\geq 0}=0^\omega$ (respectively $1^\omega)$, $(x_{4n+2})_{n\geq 0}=1^\omega$ (respectively $0^\omega)$ and $(x_{2n+1})_{n\geq 0}$ is a paperfolding word (see for instance \cite{Al92}).

We say that a factor $u$ of a paperfolding word $x$ is even (respectively odd) if $u=x_{n}x_{n+1}\cdots x_{n+|u|-1}$ with $n$ even (respectively $n$ odd). We recall the following fact:

\begin{lemma}[Lemma 2 in \cite{Al92}]\label{lempaper}  Let $x$ be a paperfolding word. If $u$ is a factor of $x$ of length $|u|\geq 7,$ then $u$ is either even or odd but not both.\end{lemma}

\begin{proposition}\label{prop:paper}
 Let $x=x_0x_1x_2\cdots $ be a paperfolding word. Then for each $n\geq 0$ and each factor $u$ of $x$ of length $4\cdot 2^{n+1},$  the cyclic class of $u$  consists of $|u|$-many distinct factors of $x.$  In particular, since $p_x(m)=4m$ for $m\geq 7$ (see \cite{Al92}),  we have $c_x(4\cdot 2^{n+1})=4$ for each $n\geq 0.$
\end{proposition}

\begin{proof} We show by induction on $n$ that for each paperfolding word $x$ and for each factor $u$ of $x$ of length $4\cdot 2^{n+1},$  the cyclic class of $u$  consists of $|u|$-many distinct factors of $x.$ The case $n=0$ is verified by direct inspection. 
For the inductive step, let $x=x_0x_1x_2\cdots $ be a paperfolding word, and let $u$ be a factor of $x$ with $|u|=4\cdot 2^{n+1}.$
We show that $x$ contains $|u|$-many distinct factors each of which is conjugate to $u.$ 
Without loss of generality we may assume $x_0=0.$  Also without loss of generality, we may suppose that $u$ is an even factor of $x.$ In fact, if $u$ is an odd factor of $x$ ending in some letter $a\in \{0,1\},$ then $u'=aua^{-1}$ is an even factor of $x$ conjugate to $u.$
So suppose $u=x_{2m}x_{2m+1}\cdots x_{2m+|u|-1}$ for some $m\geq 0.$ Let $x'=x_1x_3x_5\cdots $ and  $v=x_{2m+1}x_{2m+3}\cdots x_{2m+|u|-1}.$ Then $v$ is a factor of the paperfolding word $x'$ and $|v|=4\cdot 2^n.$ Thus by induction hypothesis,  the cyclic class of $v$ consists of $|v|$-many distinct factors of $x'.$  For each  conjugate $w=w_1w_2\cdots w_{|v|}$ of $v,$ if $w$ is an even factor of $x'$ then $0w_11w_2\cdots 0w_{|v|-1}1w_{|v|}$ is an even factor of $x$ conjugate to $u,$ while if  $w$ is an odd factor of $x'$ then $1w_10w_2\cdots 1w_{|v|-1}0w_{|v|}$ is an even factor of $x$ conjugate to $u.$ Thus we have  $|v|$-many distinct conjugates of $u$ each of which is an even factor of $x.$ On the other hand, if $z$ is an even factor of $x$ conjugate to  $u,$  then $a^{-1}za$ (where $a$ is the initial letter of $z)$ is an odd factor of $x$ conjugate to $u.$ Thus we also have $|v|$-many distinct conjugates of $u$ each of which is an odd factor of $x.$  Since $|u|\geq 7$ it follows from Lemma~\ref{lempaper} that the cyclic class of $u$  contains $2|v|=|u|$ distinct elements each of which is a factor of $x.$ \end{proof}

\subsection{Thue-Morse Word}

Let
\[
t=t_0t_1t_2\cdots =011010011001011010010110\cdots
\]
be the Thue-Morse word, i.e., the fixed point beginning with $0$ of the uniform substitution $\mu:0\mapsto 01,1\mapsto 10$. We prove that $\liminf_{n\to \infty} c_t(n)=+\infty$. It is known that $t$ is {\it overlap-free}, that is, does not contain as a factor any word of the form $avava$, where $a\in \{0,1\}$ and $v\in\{0,1\}^{*}$.

For every $n\geq 4$, the factors of length $n$ of $t$ belong to two disjoint sets: those which only occur at even positions in $t$, and those which only occur at odd positions in $t$. In fact, the factors of length $4$ of $t$ are partitioned according to  $\{0101, 0110, 1001, 1010\}$ which only occur at even positions, and $\{0010, 0011, 0100, 1011, 1100, 1101\}$ which only occur at odd positions. Except for $0101$ and $1010$ all other factors of length $4$ contain an occurrence of $00$ or $11$ and hence are decodable under $\mu.$ On the other hand, since $t$ is overlap-free, every occurrence of $0101$ in $t$ is the image under $\mu$ of an earlier occurrence of $00,$ and similarly every occurrence of $1010$ in $t$ is the image under $\mu$ of an earlier occurrence of $11.$

Let $p(n)$ be the factor complexity function of $t$. It is known \cite[Proposition 4.3]{Br89}, that for every $n\geq 2$ one has $p(2n)=p(n)+p(n+1)$ and $p(2n+1)=2p(n+1)$. Let  $f_{aa}(n)$ (resp.~$f_{ab}(n)$) denote the number of factors of $t$ of length $n$ which begin and end with the same letter (resp.~with different letters).

\begin{lemma}\label{lem:tm}
 For every $n\geq 2$, one has $f_{aa}(n)\geq  p(n)/3$ and $f_{ab}(n)\geq  p(n)/3$.
\end{lemma}

\begin{proof}
By induction on $n$. The cases $n=2,3$ are readily verified. We now suppose $n\geq 2$ and prove the statement for $2n$ and $2n+1$. Let us first consider $f_{aa}(2n)$. The factors of length $2n$ of $t$ belong to two disjoint sets: those that begin at even positions in $t$, which are images of factors of $t$ of length $n$ under $\mu$, and those that begin at odd positions in $t$. The factors in the first group are in bijection with the factors of $t$ of length $n$ that begin and end with different letters, since the former are the images under $\mu$ of the latter. The factors in the second group are in bijection with the factors of $t$ of length $n+1$ that begin and end with different letters, since the former are obtained by deleting the first and the last letter from the images under $\mu$ of the latter. So, $f_{aa}(2n)=f_{ab}(n)+f_{ab}(n+1)$. By the inductive hypothesis we have $f_{aa}(2n)\geq p(n)/3+p(n+1)/3=p(2n)/3$.
Let us now consider $f_{ab}(2n)$. Arguing similarly as in the previous case, we have $f_{ab}(2n)=f_{aa}(n)+f_{aa}(n+1)$ and therefore by the inductive hypothesis we have $f_{ab}(2n)\geq p(n)/3+p(n+1)/3=p(2n)/3$.

Consider now $f_{aa}(2n+1)$. The factors of length $2n+1$ of $t$ belong to two disjoint sets: those that begin at even positions in $t$, which are images of factors of $t$ of length $n$ under $\mu$ followed by one letter, and those that begin at odd positions in $t$, which are images of factors of $t$ of length $n$ under $\mu$ preceded by one letter. The factors in the first group are in bijection with the factors of $t$ of length $n+1$ that begin and end with the same letter, since the former are obtained by deleting the last letter from  the images under $\mu$ of the latter. Also the factors in the second group  are in bijection with the factors of $t$ of length $n+1$ that begin and end with the same letter, since the former are obtained by deleting the first letter from  the images under $\mu$ of the latter. So, $f_{aa}(2n+1)=2f_{aa}(n+1)$. By the inductive hypothesis we have $f_{aa}(2n+1)\geq 2p(n+1)/3=p(2n+1)/3$.
Finally, consider $f_{ab}(2n+1)$. Arguing similarly as in the previous case, we get $f_{ab}(2n+1)=2f_{ab}(n+1)$. By the inductive hypothesis we have $f_{ab}(2n+1)\geq 2p(n+1)/3=p(2n+1)/3$.
\end{proof}

Since $p(n)\geq 3(n-1)$ for every $n$ \cite[Corollary 4.5]{DelVa89}, we obtain:

\begin{corollary}\label{cor:tm}
  For every $n\geq 2$, one has $f_{aa}(n)\geq  n-1$ and $f_{ab}(n)\geq  n-1$.
\end{corollary}

\begin{proposition}\label{prop:tm}
 Let $t$ be the Thue-Morse word. Then $\liminf_{n\to \infty} c_t(n)=+\infty$.
\end{proposition}

\begin{proof}
We show that for each $n\geq 4,$ there exist at least $n$ factors of $t$ of length $2n$ each of which has no other factor of $t$ in its conjugacy class, and at least  $n$ factors of $t$ of length $2n+1$ each of which has at most $3$ other factors of $t$ in its conjugacy class. This of course implies  $\liminf_{n\to \infty} c_t(n)=+\infty$.

Fix $n\geq 4.$ By Corollary \ref{cor:tm}, there are at least $n$ factors of length $n+1$ which begin and end with different letters. Applying $\mu$, we obtain at least $n$ factors of length $2n+2$ which begin  with $ab$ and end with $ba$, where $\{a,b\}=\{0,1\}$.  By deleting the first and the last letter of each we obtain at least $n$ factors $v$ of length $2n$ which begin and end with the same letter and occur in $t$ at odd positions.
We claim that each such $v$ is unique in its conjugacy class. In fact, let $v'\neq v$ be  conjugate to $v.$ Then we can write $v'=ybbx$ and $v=bxyb$, for some words $x,y\in \{0,1\}^*$ such that $|x|+|y|\geq 6.$ If $v'$ is a factor of $t$, then  $bx$ occurs  both at an odd position (since it is a prefix of $v$) and at an even position (since $bb$ can only  occur in $t$ at an odd position). Hence $|bx|\leq 3$.
Moreover, also $yb$ occurs in $t$ both at an odd and at an even position, whence $|yb|\leq 3$, a contradiction.

Next we consider odd lengths.  By Corollary \ref{cor:tm}, there exist at least $n$ factors of length $n+1$ which begin and end with the same letter. As above, applying $\mu$ we obtain at least $n$ factors of length $2n+2$ which begin and end with $ab$, where $\{a,b\}=\{0,1\}.$  By deleting the first letter from each, we obtain at least $n$ factors of $t$ of length $2n+1$ which begin with $b$, end with $ab$, and occur in $t$ at odd positions. We claim that each such factor $v=bzb$, $|z|\geq 7,$ admits at most $3$ other factors of $t$ in its conjugacy class. Indeed, let $v'\neq v$ be conjugate to $v$. Then $v'$ can be written as $v'=z'bbx$, for a (possibly empty) prefix $x$ of $z$. If $v'$ is a factor of $t$, then, as above, $bx$ occurs  both at an odd and an even position. Hence $|bx|\leq 3$, and as $v'$ is entirely defined by $v$ and $|bx|$, we conclude that there are at most $3$ other factors of $t$ in the conjugacy class of $v$.
\end{proof}

\section{Acknowledgements}

We thank the anonymous referees for their careful reading of the manuscript and for providing us with many helpful comments and suggestions.  The second and third author acknowledge the support of the PRIN 2010/2011 project ``Automi e Linguaggi Formali: Aspetti Matematici e Applicativi'' of the Italian Ministry of Education (MIUR). 


\end{document}